\documentclass[pra,aps,twocolumn,superscriptaddress, longbibliography, notitlepage]{revtex4-1}
\usepackage[colorlinks=true, citecolor=blue, urlcolor=blue ]{hyperref}
\usepackage{graphicx}
\usepackage{dcolumn}
\usepackage{bm}
\usepackage{ragged2e}
\usepackage{amsmath,amsfonts}
\usepackage[english]{babel}
\usepackage{amsthm}
\usepackage{hyperref}
\theoremstyle{plain}

\newtheorem{theorem}{Theorem}

\usepackage[caption=false]{subfig}
\usepackage{xcolor}
\begin{document}

\title{Sequential detection of genuine multipartite entanglement is unbounded for entire hierarchy of number of qubits recycled}



\author{Chirag Srivastava}
\affiliation{Harish-Chandra Research Institute, A CI of Homi Bhabha National Institute, Chhatnag Road, Jhunsi, Prayagraj 211 019, India}
\author{Mahasweta Pandit}
\affiliation{Institute of Theoretical Physics and Astrophysics, Faculty of Mathematics, Physics and Informatics, University of Gda\'nsk, 80-308 Gda\'nsk, Poland}
\author{Ujjwal Sen}
\affiliation{Harish-Chandra Research Institute, A CI of Homi Bhabha National Institute, Chhatnag Road, Jhunsi, Prayagraj 211 019, India}



\begin{abstract}
Experimental detection of entanglement certainly disturbs the underlying shared state. It is possible that the entanglement content of the system is lost in the process of its detection. This observation has led to the study of sequential detection properties of various quantum correlations. 
Here, we investigate the sequential detection of genuinely multipartite entanglement of quantum systems composed of an arbitrary number of qubits. In order to detect genuine multipartite entanglement sequentially, observers can  recycle any fixed subset of all the qubits, thus leading to a hierarchy of scenarios, categorized according to the number of qubits which are recycled by the observers. We show that the sequential detection of genuine multipartite entanglement, for every scenario in the hierarchy, leads to an unboundedly long sequence. This is shown to be possible if the initial state shared among the observers is the multipartite generalized Greenberger-Horne-Zeilinger state and is a class of mixed states. A comparison among different hierarchical scenarios is drawn based on the number of sequential detections of genuine multipartite entanglement for a specific measurement strategy employed by observers.
\end{abstract}

\maketitle

\section{Introduction}

Entanglement \cite{horodecki09,Guhne09} is one of the most important signatures of non-classicality, and the notion of genuine multipartite entanglement (GME) becomes even conceptually engaging with an exponentially increasing number of state parameters in the fray with the number of parties sharing the state. Constructing entanglement witnessing tools in the multipartite scenario is therefore crucial for characterization of such correlations as well as to find novel ways to utilise them to extract quantum advantages in different quantum information processing and communication tasks \cite{Ekert91, bb84,Wiesner92,Bennett93,Zukowski93, Bose98,Mayers98,Barrett05,Acin07, Colbeck12,Colbeck111,Colbeck112,Pironio10,Cleve99,Zhao04,Gour07,Raissi18,Andrea21}. Witnessing GME remains a very demanding task \cite{Zuku03,Weinfurter04,GISIN19981,Scarani02,guhneex03,Haffner05,Krammer09,Li10,Jungnitsch11,Brunner12,Sperling13,Knips16,Singh18,Friis19,Zhou19,Teh19,Amaro20}.  

Recyclability of non-classical correlations refers to the concept of witnessing them in a sequential manner by conducting local measurements on one (or more) of the sub-systems and passing down the post-measurement sub-system(s) to the next observer. The task is to sequentially witness correlations, e.g., the violation of the Clauser-Horne-Shimony-Holt (CHSH) \cite{Clauser69} Bell inequality \cite{Silva15,Brown20} or entanglement \cite{srivastava21,pandit22,Srivastava22}.
For related works, see e.g. \cite{Mal16,Bera18,Das19,Saha19,Maity20,sriv21,roy20,cabello20,ren21,Fei21,zhu21,Das21}.
For experiments on these directions, see e.g. \cite{Schiavon17,Hu18,Vallone20,Choi20,Feng20}.
Using bipartite systems where only one of the subsystems is being recycled, both CHSH Bell-nonlocal correlation \cite{Brown20} and entanglement (for CHSH local states as well  \cite{srivastava21}) can be witnessed $arbitrarily$ many times. 
For the case where both the subsystems are recycled, it has been conjectured that not more than a single pair of observers can violate the CHSH Bell-nonlocality \cite{Cheng21,Hall21}, whereas the number of times entanglement can be witnessed is unbounded \cite{pandit22}.
Recently, it has also been reported that if only a single subsystem is recycled of an arbitrary $N$-partite GME state, then the number of sequential detections of GME is also unbounded \cite{Srivastava22}. 

In the current work, we investigate a general scenario for sequential detection of genuinely multipartite entanglement of a multi-qubit system.
Specifically, a set of spatially separated observers begin with an initial genuinely multipartite entangled state such that each observer possesses a single qubit. 
The genuinely multipartite entangled state is then subjected to measurements by each observer to witness GME of the shared state. Any fixed subset of the post-measurement qubits are then recycled, i.e., passed on to a next set of observers for another set of local measurements. The same subset may again be passed on to another subsequent set of observers for another set of local measurements, and so on. One can christen such a scheme of recycling of qubits as \emph{hierarchical recycling}, where the cardinality of the subset of observers executing the recycling forms the different levels of hierarchy. Utilizing specially engineered measurement strategies, we show that the number of sequential GME witnessing is unbounded for hierarchical recycling of qubits, for any cardinality in the hierarchy. We prove the statement for the multipartite Greenberger-Horne-Zeilinger \cite{GHZ,Mermin,Dik99} (GHZ) state, multipartite generalized GHZ states, and a class of mixed genuinely multipartite entangled states. 
We also compare different scenarios of hierarchical recycling based on the number of sequential detections of GME, where a specific measurement strategy is employed by the observers.   
We find that there exists a measurement strategy for which the total number of instances of witnessing GME is independent of the total number of qubits if only a single qubit of multipatite GHZ state is recycled.  
\section{Scenario}
Consider $N(\geq 3)$ spatially separated observers, sharing an $N$-qubit genuinely multiparty entangled state, $\rho_1$, such that each party has a single qubit. Let us suppose that a fixed set of $N_0 (=1,2,\ldots N)$ qubits can be recycled by sequential observers. The observers, possessing the recyclable qubits, perform measurements  on their subsystems and transfer them to the next set of $N_0$ observers,  who act on their qubits, and pass them to subsequent set of $N_0$ observers, and so on. Let an observer, possessing the $m^\text{th}$ qubit which has been acted on by $k-1$  observers earlier, is denoted as  $O_k^m$.  The task of any $k^{th}$ set of sequential observers (set of all observers who will act on their qubit the $k^{\text{th}}$ time) is to detect genuine multiparty entanglement with the rest of the $N-N_0$ observers who act only once on their qubits. 
It is also assumed that the sequential sets of observers act independently of each other, i.e., the information about which outcome has ``clicked'' in the measurements performed by a given set of observers are not passed to the subsequent set of observers.
Note that $N_0$ can be any integer from 1 to $N$, which captures a hierarchy of scenarios starting from the case where only one fixed qubit can be manipulated and recycled by sequential observers, to the case where all of the qubits are recycled by sequential observers. We term the collection of all such recycling scenarios as \emph{hierarchical recycling}.

In order to detect genuine multipartite entanglement, local observers may employ the method of genuine multipartite entanglement witnesses \cite{Zuku03,Weinfurter04,GISIN19981,Scarani02,guhneex03,Toth05}. In this method, the expectation  value of the witness operator is non-negative for all biseparable states and negative for at least one genuine multipartite entangled states. In \cite{Toth05}, it is shown that there exist genuine multiparty entanglement witness operators such that only two measurement settings per local observer, with dichotomic outcomes, are sufficient to detect the genuinely multiparty entangled states of an arbitrary long array of qubits.
Let the positive operator-valued measure (POVM) corresponding to the $i^{\text{th}}$ measurement setting, employed on the $m^{\text{th}}$ qubit for $m=1,2,\ldots N_0$ by the $k^{\text{th}}$ set of sequential observers, be denoted by $\{E^{k}_{i_m,a_m}\}_{a_m}$, where $i_m=\{x,z\}$ denotes the two measurement settings for every site, $a_m=\{+,-\}$ denotes set of two outcomes for every setting, $\{E^k_{i_m,a_m}\}_{a_m}$ are positive operators, and $\sum_{a_m}E^k_{i_m,a_m}=\mathbb{I}_2$ (identity on the Hilbert space of dimension 2).  
Let $\rho_k$ be the $N$-qubit state shared by the $k^{\text{th}}$ set of sequential observers along with the rest of $N-N_0$ (non-sequential) observers. Then the post-measurement state after the measurements by $k^{th}$ set of sequential observers can be expressed using the von Neumann-L\"uder's rule \cite{Busch} as
\begin{equation}\label{salt}
    \rho_{k+1}=\frac{1}{2^{N_0}}\sum_{I,A} \left(\bigotimes_{m=1}^{N_0} E^{k}_{i_{m},a_{m}}\right)^{\frac{1}{2}}
    \rho_k \left(\bigotimes_{m=1}^{N_0} E^{k}_{{i_{m},a_{m}}}\right)^{\frac{1}{2}},
\end{equation}
where $I = \{i_1,i_2,\ldots, i_{N_0}\}$ is the set of all measurement settings of the $N_0$ observers performing sequential measurements, and $A = \{a_1,a_2,\ldots, a_{N_0}\}$ is the set of all outcomes of the measurements. 
\begin{figure}
    \includegraphics[width =0.48\textwidth]{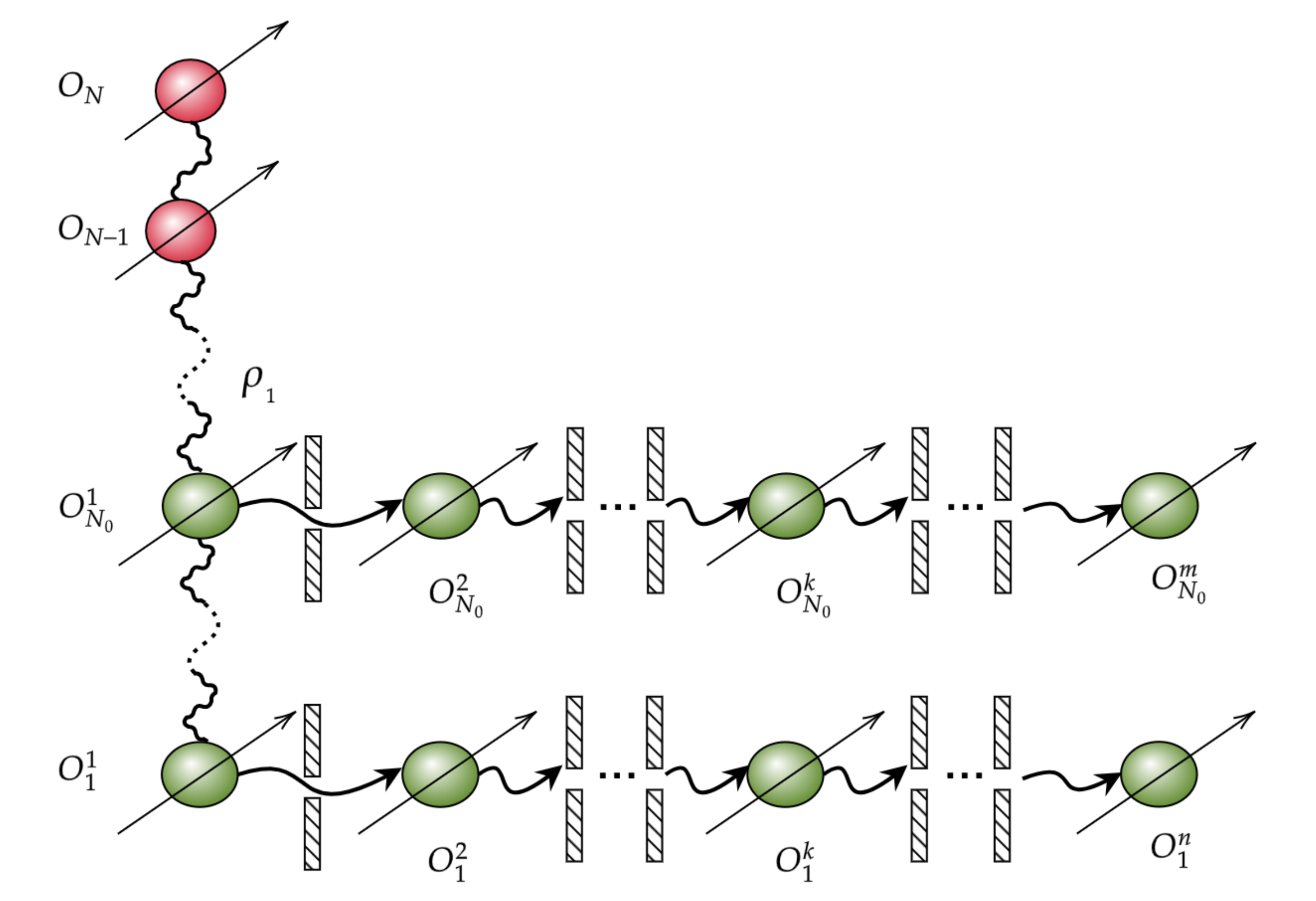}
     \caption{
      Schematic set-up for sequential detection of genuine multipartite entanglement for hierarchical recycling. $N$ spatially separated parties share an $N$-qubit quantum state. Any fixed $N_0~(=1,2,\ldots,N)$ qubits are recycled by sequential observers to detect genuine multipartite entanglement with the rest of $N-N_0$ non-sequential observers. The question is how many such sequential detections of genuine multipartite entanglement are possible for various values of $N_0$, for a given shared state.
     }
          \label{fig:kataar} 
\end{figure}
The summation over all measurement settings, $I$, is the consequence of independence of the observers acting sequentially, whereas the prefactor $\frac{1}{2^{N_0}}$ accounts for the unbiasedness of the measurement settings employed by the observers.
Note that in Eq. \eqref{salt}, it is assumed, without loss of generality, that the set of sequential observers act on the first $N_0$ qubits, so the operator identity acts on the rest of $N-N_0$ qubits. See Fig. \ref{fig:kataar}.






\section{Arbitrarily long sequential detection of genuine multipartite entanglement}
In this section, we show that the number of sequential detections of the genuine multipartite entanglement can be unbounded under the hierarchical recycling scenario, for every member of the hierarchy. For this purpose, we first consider that an $N$-partite GHZ state, 
\begin{equation}
    |\text{GHZ}_N\rangle=\frac{1}{\sqrt{2}}\left(|00\ldots0\rangle_N+|11\ldots1\rangle_N\right), \nonumber
\end{equation}
shared by $N$ spatially separated parties. Therefore, $\rho_1=|\text{GHZ}_N\rangle\langle \text{GHZ}_N|$,  where $|0\rangle$ and $|1\rangle$ are the eigenstates of Pauli operator, $\sigma_z$. We remember that $\rho_1$ denotes the initial multiparty shared state.
\subsection*{Measurement strategies, GME witnesses, and post-measurement state}
A genuine multiparty entanglement witness for detecting the $N$-partite GHZ state, for any $N$, with only two measurement settings per local observer is provided in \cite{Toth05}. 
In order to increase the number of sequential observers detecting these resources, it is essential to use a modified witness operator such that the sharp projective measurements are replaced by unsharp POVMs. 

Suppose that the $k^{\text{th}}$ set of sequential observers apply the following POVMs on any $m^{\text{th}}$ qubit for $m=1,2, \ldots,N_0$:
\begin{eqnarray}\label{sodium}
       E^k_{x,\pm}&=&\frac{\mathbb{I}_2 \pm \lambda_k \sigma_x}{2}, \nonumber \\
       E^k_{z,\pm}&=&\frac{\mathbb{I}_2 \pm \sigma_z}{2},
\end{eqnarray}
where $0 \leq \lambda_k\leq 1$ is the sharpness parameter of the measurement $\{E^k_{x,\pm}\}$. The case corresponding to $\lambda_k=1$ is a sharp projective measurement, whereas $\lambda_k=0$ corresponds to the trivial measurement where the post-measured state is completely undisturbed. The observers, acting only once on rest of the $N-N_0$ qubits, perform the following measurements on their respective qubits:
\begin{eqnarray}\label{chlorine}
       P_{x,\pm}&=&\frac{\mathbb{I}_2 \pm \sigma_x}{2}, \nonumber \\
       P_{z,\pm}&=&\frac{\mathbb{I}_2 \pm \sigma_z}{2}.
\end{eqnarray}
Consider now a GME witness operator, which is a modified version of the witness operator provided in \cite{Toth05}, used by the $k^\text{th}$ set of sequential observers and the $N-N_0$ non-sequential observers, given by
\begin{equation}\label{kapda}
    W_k=3\mathbb{I}-2\left[\frac{\mathbb{I}+\lambda_k^{N_0}S_1}{2}+ \prod_{m=2}^N\frac{\mathbb{I}+S_m}{2}\right],
\end{equation}
where $\mathbb{I}$ is the identity operator on $2^N$-dimensional Hilbert space and the operators $S_m$ are given as
\begin{eqnarray}
       S_1&=&\prod_{m=1}^N \sigma_x^{(m)}  \nonumber \\
       S_m&=&\sigma^{(m-1)}_z  \sigma^{(m)}_z, ~~m=2,3,\ldots,N,
   \end{eqnarray}
  where the superscript, $(m)$, over any operator indicates that it acts on the $m^{\text{th}}$ qubit.
The expectation value of the witness operator, $\langle W_k \rangle$, is non-negative for all biseparable states, which easily follows from the proof given in \cite{sriv22}, since $0\leq \lambda^{N_0}_k \leq 1$. Whereas, $\langle W_k \rangle_{\rho_1} < 0$ for $\lambda_k>0$, since
$\langle  S_m \rangle_{\rho_1}=1$.
Note that the measurements, defined in \eqref{sodium} and \eqref{chlorine}, are sufficient to evaluate  $\langle W_k \rangle$ over any $N$-qubit state. 
  

The updated state shared by the $k^{\text{th}}$ set of sequential observers and the non-sequential observers, due to the measurement strategy adopted by the $(k-1)^{\text{th}}$ set of sequential observers is given by,
\begin{widetext}
\begin{eqnarray}\label{choir} \rho_k&=&\frac{1}{2^{2N_0}}\sum_I\sum_{s=0}^{N_0} \sum_{\phi=1}^{\tbinom{N_0}{s}} \Xi^{N_0}_{\phi,s,I} \mathcal{P}_\phi[\sigma_{i_1}\otimes\sigma_{i_2}\otimes\ldots\otimes\sigma_{i_s}\otimes\mathbb{I}_{2^{N_0-s}}]~\rho_{k-1}~\mathcal{P}_\phi[\sigma_{i_1}\otimes\sigma_{i_2}\otimes\ldots\otimes\sigma_{i_s}\otimes\mathbb{I}_{2^{N_0-s}}], \end{eqnarray}
where $\{\mathcal{P}_\phi[\sigma_{i_1}\otimes\sigma_{i_2}\otimes\ldots\otimes\sigma_{i_s}\otimes\mathbb{I}_{2^{N_0-s}}]\}_{\phi=1}^{\binom{N_0}{s}}$ is the set of all possible instances of choosing $s$ sites among $N_0$, and planting the corresponding Pauli operators there and identity in the remaining $N_0-s$ sites. Also,
\begin{equation}
\Xi^{N_0}_{1,s,I}=(1-\Lambda^{i_1}_{k-1})(1-\Lambda^{i_2}_{k-1})\ldots(1-\Lambda^{i_s}_{k-1})(1+\Lambda^{i_{s+1}}_{k-1})\ldots(1+\Lambda^{i_{N_0-1}}_{k-1})(1+\Lambda^{i_{N_0}}_{k-1}) ,\nonumber 
\end{equation}\end{widetext}
such that the coefficient is $(1-\Lambda^{i_m}_{k-1})$ for a Pauli operator acting on the $m^{\text{th}}$ site in  Eq.~\eqref{choir}, and it is $(1+\Lambda^{i_m}_{k-1})$ for an identity. And similarly for other values of \(\phi\).
Furthermore, $\Lambda^x_k=\sqrt{1-\lambda^2_k}\equiv \Lambda_k$ and $\Lambda^z_k=0$ for any $k \in \mathbb{N}$ (the set of natural numbers).
The updated state in Eq. \eqref{choir} can be obtained using Eq.~\eqref{salt} and the following identity:
\begin{equation}
\begin{split}
    &\left(\frac{\mathbb{I}_2\pm \lambda_k\sigma_i}{2}\right)^{\frac{1}{2}} = \\ &\frac{\left(\sqrt{1+\lambda
_k}+\sqrt{1-\lambda_k}\right)\mathbb{I}_2\pm\left(\sqrt{1+\lambda_k}-\sqrt{1-\lambda_k}\right)\sigma_i}{2\sqrt{2}}.
\end{split}
\end{equation}

\subsection*{Sequential detection of GME using hierarchical recycling of multipartite GHZ state}
Let us define a set of $N$-qubit operators, $S^{q}_{2\theta,t}$, as the tensor product of $2\theta$ $\sigma_z$ operators and identity operators, where $t$ is the number of $\sigma_z$ operators acting on the space of $N-N_0$ qubits possessed by the non-sequential observers, with
$\theta=1,2,\ldots,\lfloor\frac{N}{2}\rfloor$ and $t=\tau_s,\tau_s+1,\ldots, \tau_l$, where $\tau_s = \max\{0,2\theta-N_0\}$ and $\tau_l=\min\{2\theta, N-N_0\}$.  The index $q$ is used to denote different operators arising from  the different arrangements of the tensor products of $\sigma_z$ and identity operators. 
For a given $2\theta$ and $t$, the number of possible $q$ is $\binom{N_0}{2\theta-t} \binom{N-N_0}{t}$.
Notice that the witness operator in Eq.~\eqref{kapda} is a linear combination of operators, $\mathbb{I}$, $S_1$, and all possible $S^q_{2\theta,t},$ viz.,
  \begin{equation}\label{yaad}
      W_k=\left(2-\frac{1}{2^{N-2}}\right)\mathbb{I}-\lambda_k^{N_0}S_1-\frac{1}{2^{N-2}}\sum_{\theta,t,q} S_{2\theta,t}^q
  \end{equation}
Thus, the expectation value of the witness operator, $W_k$, can be obtained if the expectation values of the operators $S_1$ and $S^q_{2\theta,t}$ are obtained. 
It is possible to express the expectation values $\langle S_1\rangle_{\rho_k}$ and $\langle S^q_{2\theta,t}\rangle_{\rho_k}$ in terms of $\langle S_1\rangle_{\rho_1}$ and $\langle S^q_{2\theta,t}\rangle_{\rho_1}$, respectively, viz.,
\begin{eqnarray}\label{dastan}
    \langle S_1\rangle_{\rho_k}&=&\left(\frac{1}{2^{N_0}}\right)^{k-1} \langle S_1\rangle_{\rho_1}, \nonumber \\
    \langle S^q_{2\theta,t}\rangle_{\rho_k}&=&\prod_{j=1}^{k-1}\left(\frac{1+\Lambda_j}{2}\right)^{2\theta-t}\langle S^q_{2\theta,t}\rangle_{\rho_1}.
\end{eqnarray}
Now the condition that the first set of observers in the sequence will detect GME is given by $\langle W_1 \rangle_{\rho_1}<0$
\begin{equation}\label{muhurt}
    \implies\lambda^{N_0}_1>0,
\end{equation}
since  $\langle S_1\rangle_{\rho_1}=\langle S^q_{2\theta,t}\rangle_{\rho_1}=1$, as $\rho_1=|\text{GHZ}_N\rangle\langle \text{GHZ}_N|$.
Expectation values given in Eq. \eqref{dastan}  can be used to determine $\langle W_k\rangle_{\rho_k}$ and the condition for detecting GME by the $k^{th}$ set of sequential observers, for any $k$, is given by $\langle W_k \rangle_{\rho_k}<0$,
\begin{eqnarray}\label{lamha}
\implies &\lambda^{N_0}_k&>\frac{2^{N_0 (k-1)}}{2^{N-2}}\Bigg[2^{N-1}-1  \nonumber \\ &-&\sum_{\theta,t}\tbinom{N_0}{2\theta-t} \tbinom{N-N_0}{t} \prod_{j=1}^{k-1} \left(\frac{1+\Lambda_j}{2}\right)^{2\theta-t}\Bigg].
\end{eqnarray}
Note that the summation over the index $q$ in the above expression is removed as the expectation values of operators $S^q_{2\theta,t}$ is independent of $q$ for the GHZ state, and equals 1.
The inequalities in \eqref{muhurt} and \eqref{lamha} can be used to define the following sequence of sharpness parameters:
  \begin{eqnarray}\label{raaz}
&\lambda^{N_0}_k& :=
      (1+\epsilon) \frac{2^{N_0 (k-1)}}{2^{N-2}}\Bigg[2^{N-1}-1 \nonumber \\ &-&\sum_{\theta,t}\tbinom{N_0}{2\theta-t} \tbinom{N-N_0}{t}\prod_{j=1}^{k-1} \left(\frac{1+\Lambda_j}{2}\right)^{2\theta-t}\Bigg],  
\end{eqnarray}
if $\lambda^{N_0}_{k-1}\in (0,1)$ and undefined otherwise,
with $\lambda^{N_0}_1>0$ and $\epsilon>0$. Genuine multipartite entanglement is detected by the  $k^{\text{th}}$ set of sequential observers if the sharpness parameter $\lambda_k \in (0,1)$.  
We can now arrive at the following theorem.
\begin{theorem}\label{raahe}
The number of sequential detections of genuine multipartite entanglement starting from a multipartite GHZ state is unbounded for entire hierarchy of number of qubits recycled. \end{theorem}
\begin{proof}
Notice that $\lambda^{N_0}_2 \to 0$ in the limit $\lambda_1 \to 0,$ because
\begin{eqnarray}
\lim_{\lambda_1\to0}\sum_{\theta,t} \tbinom{N_0}{2\theta-t} \tbinom{N-N_0}{t} \left(\frac{1+\Lambda_1}{2}\right)^{2\theta-t}&=& \sum_{\theta} \tbinom{N}{2\theta} \nonumber \\  &=& 2^{N-1}-1. \nonumber 
\end{eqnarray}
This implies $\lambda_2 \to 0$, since $\lambda^{N_0}_2 \to 0$ for an arbitrary $N_0 \leq N$, with $N$ being also arbitrary. 

Now assume that  $\lambda_j \to 0, ~ \forall j=1,2,\ldots,k-1$. It is easy to check that this implies $\lambda_k \to 0$ and thus  $\lambda_k \in (0,1),~ \forall k\in \mathbb{N}$. Hence, it is proved that an unbounded number of sequential detections of GME is possible for the entire hierarchy.
\end{proof}
\begin{figure*}
\subfloat{\includegraphics[width=0.45\textwidth]{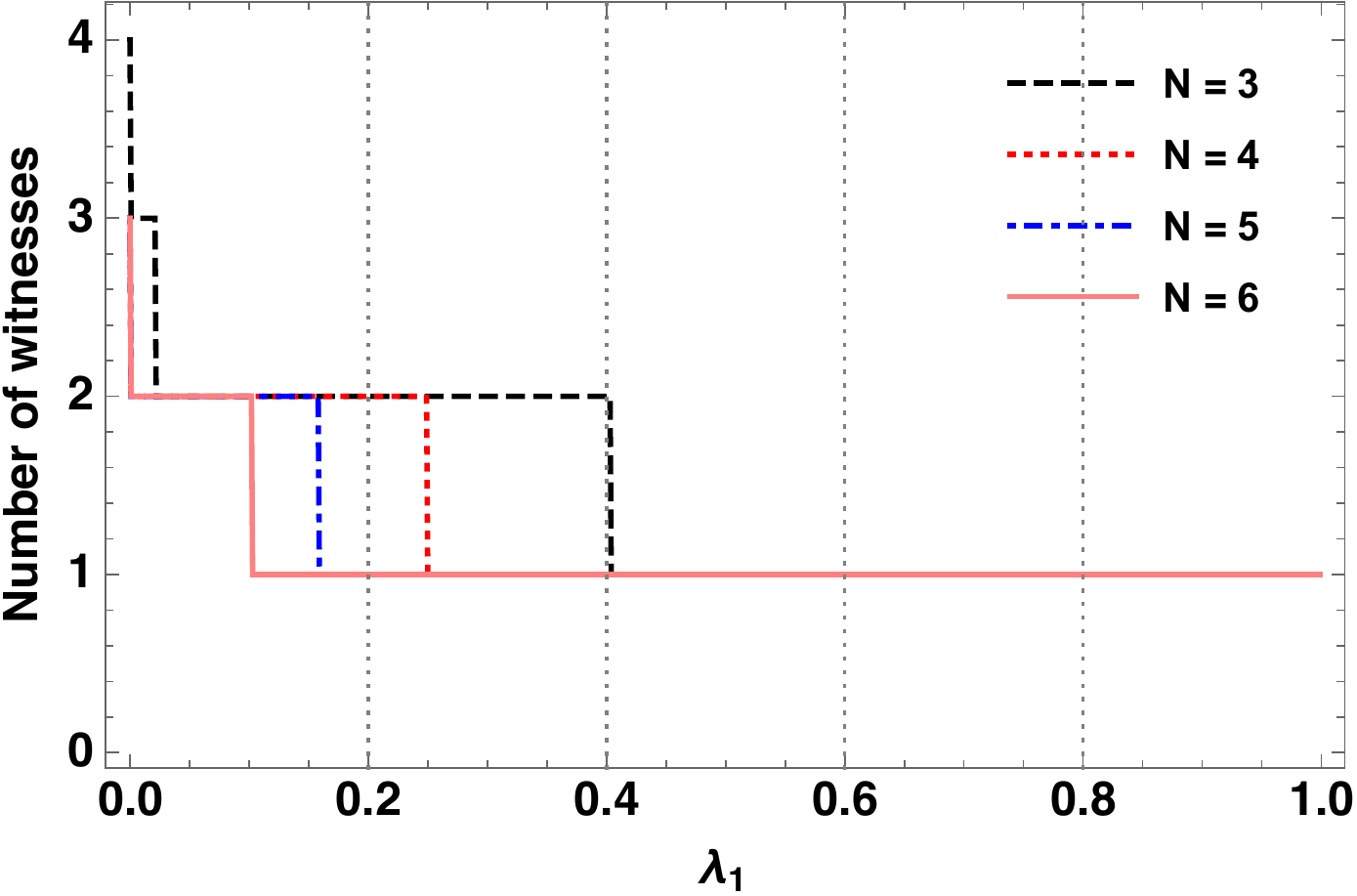}}\qquad
\subfloat{\includegraphics[width=0.45\textwidth]{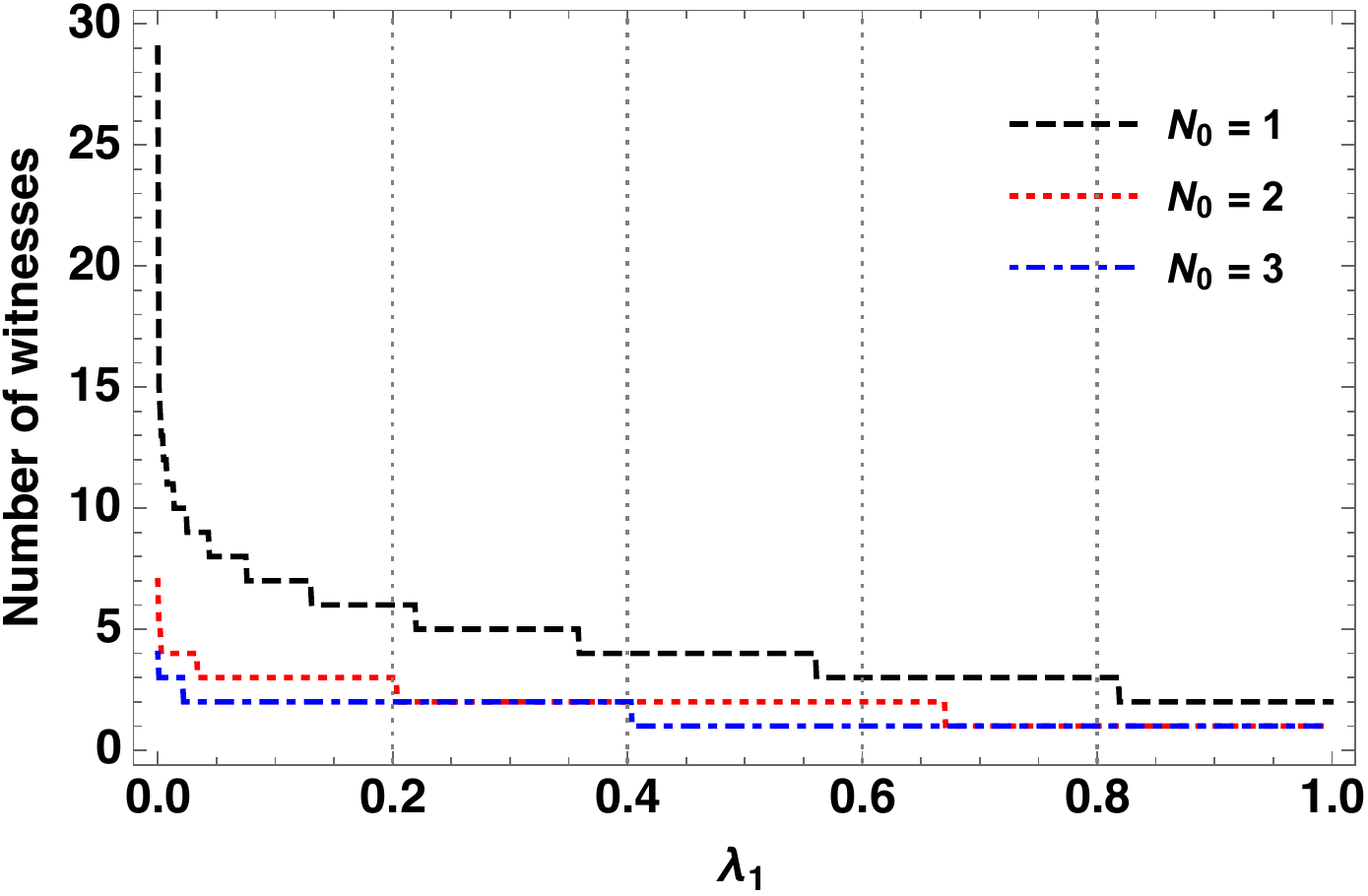}}
    \caption{Sequential witnessing of genuine multipartite entanglement and the sharpness parameter of measurements. The left panel corresponds to the scenario where all $N$ qubits are recycled. The right panel corresponds to the scenario of $N=3$, with the number of recycled qubits being varied. $\lambda_1$ denotes the sharpness parameter of the initial set of sequential observers. All quantities plotted are dimensionless.}
    \label{fig:my_label}
\end{figure*}

\subsection*{Hierarchical recycling of multipartite generalized GHZ states and a class of mixed states}
It is also possible to detect GME an unbounded number of times if the observers share the following class of $N$-partite states:
\begin{eqnarray}
    \rho_1=p_1|\text{gGHZ}_N\rangle\langle\text{gGHZ}_N|+p_2|00\ldots0\rangle\langle00\ldots&0&|\nonumber \\+(1-p_1-p_2)|11\ldots1\rangle\langle11\ldots1&|&,
\end{eqnarray}
where $p_1>0,~p_2\geq0, ~p_1+p_2 \leq 1$, and 
$$|\text{gGHZ}_N\rangle=\sqrt{a}|00\ldots0\rangle+\sqrt{1-a}|11\ldots1\rangle,$$
with $0<a\leq\frac{1}{2}.$
Notice that $\langle S_1 \rangle_{\rho_1}=p_1\sqrt{a(1-a)}$ and $\langle S^q_{2\theta,t} \rangle_{\rho_1}=1$. These imply that the sequence of sharpness parameters are given by
\begin{eqnarray}\label{raaz2}
    &\lambda^{N_0}_k& :=
      (1+\epsilon) \frac{2^{N_0 (k-1)}}{2^{N-2}p_1\sqrt{a(1-a)}}\Bigg[2^{N-1}-1 \nonumber \\ &-&\sum_{\theta,t}\tbinom{N_0}{2\theta-t} \tbinom{N-N_0}{t}\prod_{j=1}^{k-1} \left(\frac{1+\Lambda_j}{2}\right)^{2\theta-t}\Bigg],  
\end{eqnarray}
if $\lambda^{N_0}_{k-1} \in (0,1)$ and undefined otherwise, with $\lambda^{N_0}_1> 0,~\epsilon > 0$. As this sequence is just a  multiple of the sequence given in Eq. \eqref{raaz}, with the factor being independent of $N$, $N_0$, $k$, one can show in a manner similar to that in the proof of theorem~\ref{raahe} that it is possible to have $\lambda_k \in (0,1)$ for any $k\in \mathbb{N}.$ 
\section{Comparing different hierarchies of recycling}
In this section, we compare various hierarchical recycling scenarios of genuine multipartite entanglement, when the initial shared state is an $N$-party GHZ state for the specific measurement strategy adopted in this paper. Eq. \eqref{raaz} provides sequences of sharpness parameters to detect GME for all possible cases of hierarchical recycling, i.e., the number of qubits which are recycled by sequential observers  $(N_0)$ can be any number from 1 to $N$.
The sequence of sharpness parameters for $N_0=1$ reduces to
\begin{eqnarray}\label{comp}
\lambda_k=(1+\epsilon)2^{k-1}\left(1-Q_k\right), 
\end{eqnarray}
if $\lambda_{k-1}\in (0,1)$ and undefined otherwise, with $\lambda_1> 0,~\epsilon > 0$ and where $Q_k=\prod_{j=1}^{k-1}\frac{1+\Lambda_j}{2}$.
Notice that the sequence of sharpness parameters for $N_0=1$ is  independent of $N$, whereas for other values of $N_0$, the corresponding sequences, given in Eq. \eqref{raaz} are dependent on $N$ and $N_0$. Specifically, for the other extreme case, i.e., $N_0=N$, the sequence of sharpness parameters reduces to
\begin{eqnarray}\label{laptop}
    \lambda^{N}_k =
      (1+\epsilon) 2^{N (k-1)+1}\Bigg[1  -  \Bigg\{\left(\frac{1+Q_k}{2}\right)^N\nonumber \\+\left(\frac{1-Q_k}{2}\right)^N\Bigg\}\Bigg], 
\end{eqnarray}
if $\lambda^{N}_{k-1}\in (0,1)$ and undefined otherwise, with $\lambda^{N}_1> 0,~\epsilon > 0$. In the left panel of Fig. \ref{fig:my_label}, we depict the case $N_0=N$, and find that as $N$ increases, the number of sequential detections of GME either decreases or remains the same, but never increases for any given sharpness parameter used by the initial set of sequential observers. Thus, in the scenario where all spatially separated parties act sequentially on their qubits, increase in the number of qubits lowers the ``performance'' of sequential detection of GME under the considered measurement strategy. Here ``performance'' refers to the number of sequential detections of GME for a given measurement strategy. However, the number of sequential detections can be made $arbitrary$, for any value of $N$, by letting the sharpness parameter of the initial set of sequential observers tend to zero.
As the next specific case, we consider in the right panel of Fig. \ref{fig:my_label} the case of $N=3$. Here $N_0$ can be 1,2, or 3. The sequences of sharpness parameters for $N_0=1$ and 3 can be explicitly obtained from Eqs. \eqref{comp} and \eqref{laptop}, respectively. Whereas, it takes the following form for $N_0=2$:
\begin{eqnarray}
    \lambda^{2}_k=(1+\epsilon)\frac{4^{k-1}}{2}\left[3-Q_k^2-2Q_k\right],
\end{eqnarray}
 if $\lambda^{2}_{k-1}\in (0,1)$ and undefined otherwise, with $\lambda^2_1> 0,~\epsilon > 0$. It can be observed from the right panel of Fig. \ref{fig:my_label}, that with increase in $N_0$, the number of sequential detections may either decrease or remain same for any given $\lambda_1$. Thus, the performance is better for the case when the number of qubits that are being recycled, is low.



\section{Conclusion}
The method of entanglement witnesses can be employed to know about the presence of entanglement in a bipartite quantum system, and is interestingly a sought-after method to detect entanglement in experiments. Local observers perform measurements on their local systems and together evaluate the expectation values of the entanglement witness operators to guarantee the underlying entanglement. While measurements lead to the information gain about the system, they also disturb the underlying state. Thus, it may happen that witnessing entanglement cause the underlying state to lose entanglement.
This observation has led to the study of sequential detection of bipartite entanglement and Bell-nonlocal correlations.  These studies tell us about the fundamental limits on the recycled detection of  non-classical correlations and resources. 
The sequential detection scenario of the resources can potentially find advantage in quantum technologies where state preparation is a challenge.  

In this paper, we investigated sequential detection of resources in  multipartite quantum systems. Genuine multipartite entanglement is one of the most interesting non-classical correlations, and finds applications in multipartite quantum communication \cite{Cleve99,Zhao04} and quantum error correction tasks \cite{Gour07,Raissi18,Andrea21}. 
We considered a scenario where an $N$-qubit genuinely multipartite entangled state was shared among $N$ spatially separated observers such that each observer possessed a qubit.
Each party performed measurements on their respective qubit to detect genuine multiparty entanglement among all the $N$ spatially separated parties. An arbitrary subset of all the qubits was recycled after the measurements, so that sequential detection of GME is possible. The number of qubits that are recycled forms an hierarchy of sequential detection scenarios. 
We found that there exist multiparty quantum states for which the number of sequential detections of genuine multipartite entanglement is unbounded, for every member of the hierarchy. This result was shown to be valid for the case when the initial state shared among the spatially separated parties is the multipartite generalized Greenberger-Horne-Zeilinger states and is a class of mixed genuinely multipartite entangled states. Comparison among different hierarchical recycling scenarios were also drawn for the specific measurement strategy that resulted in an unbounded sequence of genuine multipartite entanglement detection. It was found that the case where only one qubit is recycled, the results are independent of the total number of qubits. For the case when all qubits were recycled, the number of sequential detections decreases with the total number of qubits.



\section*{Acknowledgement}
The research of CS was supported in part by the INFOSYS scholarship.~MP acknowledges the NCN (Poland) grant (grant number 2017/26/E/ST2/01008). The authors from Harish-Chandra Research Institute acknowledge partial support from the Department of Science and Technology, Government of India through the QuEST grant (grant number DST/ICPS/QUST/Theme-3/2019/120).

\bibliography{main}
\end{document}